\pdfoutput=1

\documentclass[11pt]{article}

\usepackage[margin=2.5cm]{geometry}

\usepackage[T1]{fontenc}
\usepackage[utf8]{inputenc}
\usepackage{lmodern}

\usepackage{amsmath,amssymb,amsthm}
\usepackage{booktabs}
\usepackage{siunitx}
\usepackage{graphicx}
\usepackage{float}
\usepackage[numbers,square]{natbib}
\makeatletter

\renewcommand\@biblabel[1]{#1.}
\makeatother
\usepackage[hypertexnames=false,colorlinks=true,allcolors=blue]{hyperref}
\usepackage{url}
\usepackage{xcolor}

\newtheorem{proposition}{Proposition}

\graphicspath{{./}}

\hypersetup{
  bookmarksdepth=3,
  pdftitle={HierX: Fast Multi-Scale Distance-Decay Interaction on Million-Node Networks},
  pdfauthor={Alexander Hellervik, Joakim Bohlin, Claes Andersson},
  pdfsubject={Hierarchical sparse-plus-correction operators for distance-decay
    interaction fields on large spatial networks},
  pdfkeywords={hierarchical methods, spatial interaction, network accessibility,
    sparse operators, distance-decay kernels, O(n log n) complexity}
}

\title{\bfseries HierX: Fast Multi-Scale Distance-Decay Interaction\\
  on Million-Node Networks}

\author{%
  Alexander Hellervik\textsuperscript{1,*} \and
  Joakim Bohlin\textsuperscript{2} \and
  Claes Andersson\textsuperscript{1}
}

\date{\today}

\begin{document}

\maketitle

\begin{center}
\small
\textsuperscript{1}Department of Physical Resource Theory,
  Chalmers University of Technology, SE-412\,96 Gothenburg, Sweden\\
\textsuperscript{2}Department of Physics,
  Chalmers University of Technology, SE-412\,96 Gothenburg, Sweden\\[4pt]
\textsuperscript{*}Corresponding author:
  \href{mailto:alexander.hellervik@chalmers.se}{alexander.hellervik@chalmers.se}\\
\href{mailto:joakim.bohlin@chalmers.se}{joakim.bohlin@chalmers.se} \quad
\href{mailto:claeand@chalmers.se}{claeand@chalmers.se}
\end{center}

\begin{abstract}
\noindent
Many models across the sciences require global distance-decay
interactions on large sparse networks. Gravity models, accessibility
measures, spatial economic models, and network influence processes all
evaluate distance-weighted potential fields: each location accumulates
contributions from every other location, weighted by a decaying
function of the shortest-path travel cost between them. Computed
directly, such aggregate fields require the dense matrix of all
pairwise network costs, which scales quadratically in time and memory.
Common approximations either discard long-range contributions or fail
when interaction is governed by network distances rather than
geometric proximity.
We introduce HierX, a hierarchical sparse-plus-correction operator for
distance-decay potential fields on networks. HierX constructs
multi-scale representative layers with explicit correction terms
ensuring each location pair contributes exactly once at the finest
available resolution. Under bounded-growth assumptions common in
spatially embedded networks, applying HierX scales as $O(n \log n)$.
Systematic benchmarks confirm quasi-linear scaling to 100{,}000 nodes.
In head-to-head comparison with distance cutoff truncation and
Nystr\"om low-rank approximation on 25{,}000-zone networks, HierX
achieves 5--9\% RMSE at a fraction of the computational work;
Nystr\"om degrades severely on steep decay kernels. Case studies
compute population-weighted accessibility on the 2.58-million-node
Great Britain driving network and the 1.77-million-node London
pedestrian network: a one-time hierarchy construction
(${\sim}1$\,hour) yields a compact reusable operator that then
evaluates each national-scale accessibility field in under one second
(${\sim}700$\,ms for Great Britain, ${\sim}150$\,ms for London).
Open-source code and worked examples are provided.
\end{abstract}


\section*{Significance Statement}
Computing global distance-decay interactions on large networks is a
bottleneck across many disciplines: the cost scales quadratically with
network size, limiting models to networks far smaller than real-world
systems. We introduce HierX, a hierarchical operator that reduces this
cost to quasi-linear, resolving nearby interactions exactly while
approximating distant ones through multi-scale aggregation. From Great
Britain's 2.58-million-node road network, HierX builds a compact
multi-scale interaction operator in about an hour; once stored, the
operator evaluates national-scale accessibility surfaces in under one
second per scenario. This enables applications beyond the reach of
quadratic methods: accessibility analysis for transport and land-use
planning, gravity-model evaluation, and catchment studies, where many
scenarios reuse one network. An open-source
implementation makes these capabilities accessible.
 
\section{Introduction}
\label{sec:introduction}

Across many domains---transportation, spatial economics, ecology,
epidemiology, network science---models of large-scale systems share a
common structure: every node interacts with every other, but the
strength of interaction decays with
distance~\cite{tobler1970computer,batty2013new}. Urban
accessibility~\cite{hansen1959accessibility,geurs2004accessibility},
gravity-type spatial interaction
models~\cite{wilson1971family,fotheringham1989spatial}, and
spatial-economic models~\cite{ahlfeldt2015economics} all depend on dense
matrices of pairwise interactions defined over network-based costs.
We focus on transportation and spatial interaction as motivating
applications, but the formulation applies to any distance-decay kernel
on a sparse graph.

In all these cases, the underlying interaction structure can be
represented uniformly as:
\begin{equation}
  h_i = \sum_j f(c_{ij}) \, x_j,
  \label{eq:interaction}
\end{equation}
where $c_{ij}$ is a network-based shortest-path cost,
$f$ is a distance-decay function, $x_j$ is an attribute such as
population, activity, or neural state, and the output $h_i$ is the
resulting aggregate potential at node~$i$---in transport applications,
the accessibility of location~$i$.

Equation~\eqref{eq:interaction} is a recurring computational kernel
within these models rather than a complete representation of them.
Gravity and accessibility models may be singly or doubly constrained,
so that the interaction between $i$ and $j$ depends on competition
among all destinations, and applications such as congested network
loading require explicit paths rather than aggregate potentials. We
target the unconstrained kernel~\eqref{eq:interaction}, which appears
repeatedly as a building block---for example, inside the balancing
iterations of constrained models---and whose evaluation dominates the
cost at scale.

The main obstacle is scale. Computing and storing the full interaction
matrix $f(c_{ij})$ requires $O(n^2)$ time and memory in the number
of nodes~$n$. For realistic networks of $10^4$ to $10^6$ nodes, direct
computation remains possible but is practically limited to
supercomputing resources (terabytes of memory at $n = 10^6$), placing
it out of reach for routine modeling work. Common alternatives each have well-known
trade-offs: truncating interactions beyond a distance cutoff is simple
and efficient but discards long-range interaction mass---especially
problematic for slowly decaying kernels; low-rank and random-feature
approximations~\cite{williams2001nystrom,rahimi2007random} are
effective for positive-definite kernels in Euclidean settings but lack
guarantees for non-PSD shortest-path kernels on irregular graphs;
local kernels ($k$-nearest neighbors or adjacency-based
neighborhoods) preserve short-range structure but sacrifice global
reach. These methods can be excellent choices in their respective
regimes, but none simultaneously provides sub-quadratic computation,
global coverage, and accurate short-range interactions for network-based
distance-decay kernels.

We address this gap with a hierarchical interaction operator that
replaces the dense matrix with a multi-resolution decomposition,
achieving $O(n \log n)$ storage and computation under standard sparsity
and bounded-growth assumptions (Section~\ref{sec:complexity}).

\textbf{Intuitive overview.}
The hierarchical operator replaces the exhaustive $n^2$ pairwise
computation with a multi-resolution strategy: nearby locations are
treated individually with exact interactions, while distant locations
are grouped and interact through representatives. The network is
organized into increasingly coarse layers; short-range interactions,
which contribute most and are most sensitive to local structure, are
resolved exactly, while long-range interactions are approximated through
group representatives at a coarseness proportional to distance.
Correction matrices prevent double-counting across overlapping layers,
ensuring that each pair contributes once at the finest available
resolution.

\textbf{Contributions.} This paper presents an operator-level
formulation that adapts hierarchical interaction ideas to shortest-path
decay kernels on general sparse graphs. We integrate ideas from fast multipole and hierarchical matrix methods,
multi-scale spatial representations, and adaptive zoning into a unified
sparse-plus-correction operator. Our main contributions are:
\begin{enumerate}
\item \textbf{A domain-tailored hierarchical operator} for interaction
  kernels of the form~\eqref{eq:interaction} on graph-based
  shortest-path costs, expressed as
\begin{equation}
    H = \sum_k G_k^\top (F_k - \mathrm{Corr}_k) \, G_k,
    \label{eq:operator}
  \end{equation}
where $G_k$ maps nodes to representatives at layer~$k$, $F_k$ is the
  sparse interaction matrix at that layer, and $\mathrm{Corr}_k$
  subtracts coarser-layer contributions to ensure each node pair is
  counted exactly once.

\item \textbf{A practical construction of a multi-layer hierarchy on
  graphs}, including representative selection, grouping, and sparse
  storage, designed for geographically embedded networks.

\item \textbf{A complexity analysis} showing that, under standard
  sparsity and growth assumptions, application scales as $O(n \log n)$
  and local network changes can be propagated through the hierarchy
  with work proportional to the number of affected pairs times the
  hierarchy depth.

\item \textbf{An applied and reproducible formulation}, with open-source
  code and worked examples aimed at lowering the barrier to
  hierarchical interaction methods.
\end{enumerate}

Unlike hierarchical matrix and fast multipole frameworks, which rely on
geometric admissibility and low-rank block compression in Euclidean
spaces, our operator acts directly on shortest-path metrics of arbitrary
sparse graphs, targeting large-scale spatial interaction computation
without requiring geometric embedding or positive-definiteness.

\textbf{Scope and intended niche.} The operator targets the regime where
(i)~the network is large enough ($n \gtrsim 10^4$) that the quadratic
cost of dense computation becomes prohibitive,
(ii)~the decay function is shallow enough
that distant interactions carry meaningful mass, (iii)~the application
requires repeated evaluations of $H\mathbf{x}$ for varying~$\mathbf{x}$,
amortizing the build cost, and (iv)~the graph is a sparse, roughly
spatial network with non-negative edge costs.

\textbf{Paper outline.}
The remainder of the paper proceeds as follows.
Section~\ref{sec:background} reviews the hierarchical lineage from
cellular automata and adaptive zoning, and situates our approach
relative to existing hierarchical and graph-based methods.
Section~\ref{sec:methods} then introduces the hierarchical operator
itself, including the construction of representative layers and the
sparse-plus-correction formulation. Particular attention is paid to the
algorithmic choices that make the method practical on large
transport-like networks, including representative selection and grouping.
Sections~\ref{sec:complexity}--\ref{sec:experiments} analyze
computational complexity, approximation error, and empirical behavior
in realistic network settings, including head-to-head comparisons with
distance cutoff truncation and Nystr\"om approximation.
Section~\ref{sec:discussion} discusses limitations, the method's
comparative niche, and directions for future work.
 \section{Background and Lineage}
\label{sec:background}

\subsection{Hierarchical cell-space and the 2002 CA model}

This work builds on research originating in cellular automata (CA) and
complex systems. Andersson et al.~\cite{andersson2002urban} introduced a
CA\slash Markov random field model for urban settlement dynamics in
which each cell is influenced by the entire lattice, with interaction
strengths decaying with distance. A key motivation was what might be called the ``speed of
light'' problem in geographical CA (cf.~\cite{andersson2002urban}):
finite neighborhoods artificially couple spatial and temporal scales.

To overcome this, the 2002 model introduced a hierarchical
cell-space\slash mean-field renormalization scheme in which the lattice
is recursively aggregated into larger blocks carrying average states.
Interactions at larger distances are computed between aggregated blocks,
reducing the cost from $O(N^2)$ to approximately $O(N \log N)$ while
retaining fine-scale detail locally and providing a multi-scale
representation of space in which distant influences are represented more
coarsely than local ones.

\subsection{Variable-grid CA and adaptive zoning}

Several authors subsequently generalized the hierarchical cell-space
idea. White~\cite{white2006modeling} recast it as a variable-grid CA
where distant rings are represented by larger ``super-cells''; Van Vliet
et al.~\cite{vanvliet2009modeling} implemented this for urban growth
modeling. Outside the CA framework, Hagen-Zanker and
Jin~\cite{hagenzanker2012adaptive} proposed adaptive zoning for
gravity-type spatial interaction models, aggregating distant destination
zones into larger regions. They note the approach could extend beyond
Euclidean distance but do not formulate a general graph-theoretic
operator.

Variable-grid CA and adaptive zoning are thus spatial siblings of the
present work: they express the idea that spatial detail can be coarser at
longer distances, but remain tied to raster grids or zonal systems.

\subsection{Connections to existing hierarchical and graph-based methods}
\label{sec:related}

We briefly describe how the proposed operator relates to existing
hierarchical and multi-scale methods.

\paragraph{Hierarchical matrices and fast multipole methods.}
The Barnes--Hut treecode~\cite{barneshut1986} and fast multipole
methods~\cite{greengard1987fast} (FMM) approximate far-field
interactions via hierarchical grouping and series expansions, while
$\mathcal{H}$-matrices~\cite{hackbusch1999sparse,hackbusch2015hierarchical,borm2003introduction}
compress well-separated blocks in low rank; together these reduce
$N$-body interactions from $O(N^2)$ to $O(N \log N)$ or $O(N)$. Our
sparse-plus-correction form~\eqref{eq:operator} shares this near/far
decomposition, but there are important differences:
\begin{itemize}
\item We work with \textbf{shortest-path distances on general graphs},
  not with kernels defined on geometric point sets in Euclidean space.
\item We approximate far-field interactions via \textbf{group
  representatives and precomputed costs}, rather than via low-rank
  factorizations of block submatrices.
\item We focus on \textbf{domain-specific interaction kernels} common in
  transport and spatial interaction modeling, rather than on general
  elliptic PDE operators.
\item We emphasize \textbf{dynamic updates to the underlying graph}
  (adding nodes, links, or changing costs), which are less central in
  classical $\mathcal{H}$-matrix and FMM settings.
\end{itemize}
Our work is thus an applied, graph-oriented sibling of these methods
rather than a competitor to general $\mathcal{H}$-matrix or FMM
frameworks. The structural parallel is closest to $\mathcal{H}$-matrices,
which treat admissible blocks independently. The
$\mathcal{H}^2$-matrix refinement~\cite{hackbusch2015hierarchical}
achieves optimal complexity by expressing cluster bases hierarchically
(nested bases), and FMM analogously translates expansions between
levels of its tree. Our representative layers are in fact already
nested in the set-theoretic sense---the propagation step of
Section~\ref{sec:rep-selection} guarantees that each fine-layer group
lies within a coarse-layer group---but the operator does not yet
exploit this nestedness algebraically: each layer stores its
representative costs independently. Two $\mathcal{H}^2$-inspired
refinements appear promising and are discussed as future work in
Section~\ref{sec:discussion}: translating cost information between
nested layers rather than recomputing it, and replacing the hard
single-representative assignment with a weighted average (convex
combination) of several representatives, analogous to
interpolation-based and variable-order
$\mathcal{H}^2$ constructions. Conversely, the correction-matrix
mechanism---which lets overlapping scales coexist while guaranteeing
exactly-once contribution without geometric admissibility
criteria---may be of interest in continuous settings.

\paragraph{Graph Laplacians, spectral sparsification, and diffusion kernels.}
Spectral sparsification~\cite{spielman2011spectral} approximates a
graph by a sparse subgraph preserving Laplacian quadratic forms.
Separately, graph diffusion
kernels~\cite{kondor2002diffusion} define
interactions via the matrix exponential of the negative Laplacian.
Both address a different
problem from ours: our focus is on interaction sums
$h_i = \sum_j f(c_{ij}) \, x_j$ where $f$ is a distance-decay function
of shortest-path costs, not a function of the graph Laplacian.

\paragraph{Distance oracles, hierarchical routing, and hub labeling.}
Distance oracles~\cite{thorup2005approximate}, contraction
hierarchies~\cite{geisberger2012exact}, transit node
routing~\cite{bast2007transit}, and hub
labeling~\cite{abraham2011hub} (see~\cite{bast2016route} for a survey)
precompute data structures for fast single-pair shortest-path queries.
Our objective is different: we compute, for all nodes $i$, a global
interaction sum $h_i = \sum_j f(c_{ij}) \, x_j$, which is closer to
applying a full operator than answering individual queries. The two
families scale differently---contraction hierarchies excel at
one-to-one queries, whereas Dijkstra's algorithm~\cite{dijkstra1959} is
naturally one-to-many---which suggests a complementary division of labor within
our hierarchy: fine layers require many short-range one-to-many
searches, for which cutoff-limited Dijkstra is well suited, while
coarse layers require costs between relatively few, widely separated
representative pairs, exactly the regime where contraction-hierarchy
or hub-label queries excel. Such structures could thus serve as a
backend for computing the coarse-layer entries of $F_k$; we do not
pursue this connection here.

\paragraph{Graph coarsening and multi-scale graph representations.}
Hierarchical graph pooling methods such as
DiffPool~\cite{ying2018hierarchical} learn to map nodes to
progressively coarser clusters for graph classification and related
tasks. Our hierarchy superficially resembles coarsening, but we do not
operate on a reduced graph: instead, we approximate a specific
interaction operator and always map back to the full node set.
Techniques from graph coarsening could nonetheless inform alternative
representative selection strategies.

\paragraph{Kernel approximation methods.}
Random Fourier features~\cite{rahimi2007random} and Nystr\"om
approximation~\cite{williams2001nystrom} provide $O(n)$ or
$O(n \log n)$ kernel matrix approximations with theoretical guarantees.
However, Nystr\"om error bounds rely on the kernel being positive
semi-definite (PSD). Interaction kernels defined on shortest-path
costs over general graphs need not be PSD, and in our baseline
comparisons (Section~\ref{sec:baselines}) standard Nystr\"om applied
to such kernels fails on steep
decay functions, consistent with this limitation. Our approach
targets a different regime: arbitrary distance-decay functions on graph
metrics, where the graph structure itself dictates the hierarchy.

\subsection{From spatial hierarchies to general network operators}

Existing work does not generalize the hierarchical idea to arbitrary
graph structures; does not define a clean algebraic operator; and does
not address dynamic structural updates. The present method lifts the
hierarchical cell-space idea from regular grids to general sparse
networks with non-negative edge costs, defines a linear operator in
sparse-plus-correction form, and provides a conceptual framework for
dynamic updates with $O(\log n)$ propagation depth.
 \section{Methods: The Hierarchical Operator}
\label{sec:methods}

\subsection{Problem setting}

We consider an undirected weighted graph $G = (V, E, w)$ with
$|V| = n$ nodes and non-negative edge costs. A designated subset
$Z \subseteq V$ of \emph{zones} carries the activity and accessibility
values of interest; the remaining nodes serve only as routing
intermediaries. We write $|Z| = n_z$. Let $c_{ij}$ be the
shortest-path cost between zones $i$ and $j$ (computed over the full
graph), and $f$ a distance-decay function. The dense interaction matrix
is $F_{ij} = f(c_{ij})$, $i,j \in Z$. Given an activity
vector~$\mathbf{x}$, the dense accessibility is
\begin{equation}
  h_i = \sum_{j \in Z} F_{ij} \, x_j.
  \label{eq:dense-matvec}
\end{equation}
Our aim is to approximate $F$ with a hierarchical operator $H$ that is
much cheaper to evaluate while retaining exact contributions for zone
pairs whose costs are known at fine resolution. In the special case
$Z = V$, every node is a zone; all benchmarks in this paper use this
setting, which represents the worst case for the operator.

\subsection{Hierarchical layers and representatives}

We build $K$ layers. Throughout, shortest-path costs are computed with
Dijkstra's algorithm~\cite{dijkstra1959}, the classical method that
explores a network outward from a source node in order of increasing
cumulative cost; it naturally supports one-to-many searches that stop
once a cost cutoff is reached. At layer~$k$, radius $\rho_k$ defines the grouping
scale, and the overlap factor $\alpha > 1$ extends the cutoff to
$\alpha \cdot \rho_k$, controlling how far each layer's Dijkstra searches reach
and thus the resolution of the approximation.
Representatives $R_k$ are chosen so that each node belongs to exactly
one representative at layer~$k$. $G_k$ is the group-membership matrix.

At the finest layer, $R_0 = V$. At coarser layers, $|R_k|$ becomes
progressively smaller.

\subsection{Representative selection and grouping}
\label{sec:rep-selection}

Representatives and groups are constructed by greedy clustering, layer
by layer from fine to coarse, given a base radius~$\rho_0$, growth
factor $\gamma > 1$ ($\rho_{k+1} = \gamma \cdot \rho_k$), and overlap
factor $\alpha > 1$ (cutoff $\alpha \cdot \rho_k$).

At layer $k = 0$, $R_0 = V$ and each node represents itself. For
$k \ge 1$, the candidates are $R_{k-1}$ and the algorithm proceeds:
\begin{enumerate}
\item Set $\rho_k = \rho_0 \cdot \gamma^k$ and
  cutoff $\alpha \cdot \rho_k$. Initialize $R_k = \emptyset$.
\item For each candidate $u \in R_{k-1}$ (ordered by index):
  \begin{itemize}
  \item If some $r \in R_k$ satisfies $c_{ur} < \rho_k / 2$,
    assign~$u$ to the nearest such~$r$.
  \item Otherwise, promote $u$ to~$R_k$ and reassign any previously
    grouped candidate~$v$ for which $c_{vu} < c_{v,r_k(v)}$
    (\textbf{takeover}).
  \end{itemize}
\item Every member of $R_{k-1}$ now belongs to exactly one group at
  layer~$k$.
\end{enumerate}
The takeover step keeps groups compact by letting a new representative
steal nearby nodes from other groups. Because candidates are drawn from
$R_{k-1}$, $|R_k|$ shrinks rapidly with~$k$.

\paragraph{Group propagation and nesting.}
A coarse-to-fine propagation step ensures every node has a representative
at every layer: if node~$i$ maps to~$r$ at layer~$k$ and $r$ maps
to~$r'$ at layer~$k+1$, then $i$ maps to~$r'$ at layer~$k+1$. This
guarantees \emph{nested} groups---each fine-layer group lies entirely
within a coarse-layer group---which the correction matrices require.
Group membership is recorded in $G_k[r, i] = 1$.

\paragraph{Cost computation.}
For each layer~$k$, we run cutoff-limited Dijkstra searches from each
representative $r \in R_k$, storing costs from $r$ only to other
representatives $u \in R_k$ reached within $\alpha \cdot \rho_k$. These
costs populate the sparse structure $C_k$. Applying the distance-decay
function $f$ to $C_k$ yields the interaction structure $F_k$ at
layer~$k$. For the coarsest layer, all representative pairs are stored
(no cutoff).

In practice this is efficient because $|R_k|$ shrinks rapidly with~$k$,
the cutoff $\alpha \cdot \rho_k$ limits each Dijkstra search, and in
spatial networks the node count within $\alpha \cdot \rho_k$ remains
modest. We adopt this greedy algorithm for its transparency and spatial
intuition; more sophisticated coarsening schemes could be substituted.

\begin{figure}[t]
  \centering
  \includegraphics[width=\linewidth]{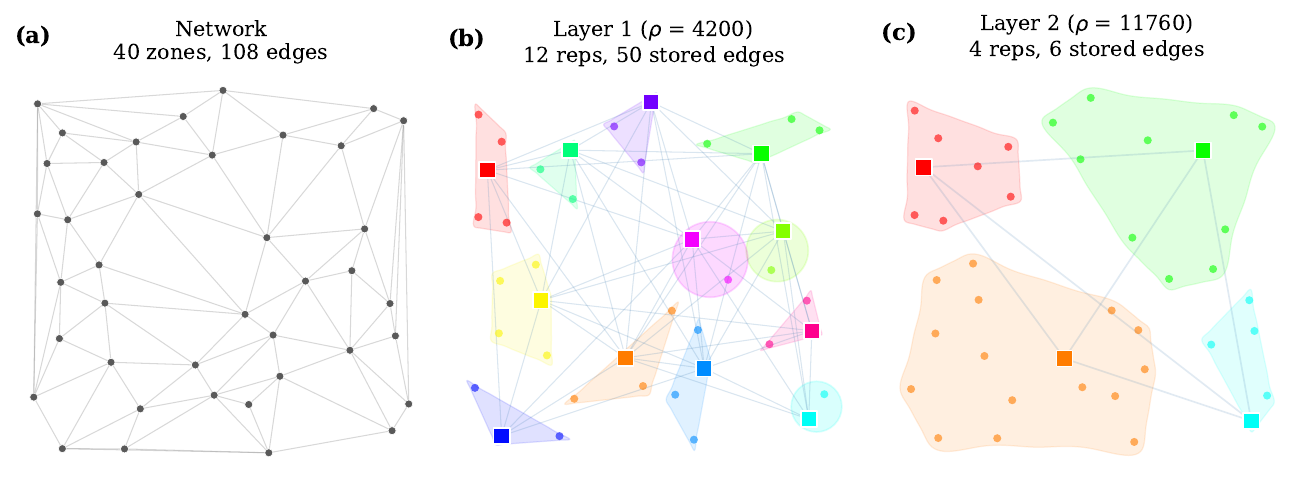}
  \caption{Hierarchy construction on a 40-zone random spatial network.
    (a)~Input graph. (b)~Layer~1 ($\rho = 4{,}200$): 12 representatives
    (squares), each heading a group (shaded region). (c)~Layer~2
    ($\rho = 11{,}760$): 4 representatives with coarser groups. Stored
    edges connect representative pairs within the cutoff
    $\alpha \cdot \rho_k$.
    \textit{Alt text:} Three panels showing hierarchy construction on a
    40-node random spatial graph, progressing from the raw network to 12
    representatives with colored groups to 4 coarser representatives with
    larger group regions.}
  \label{fig:hierarchy-construction}
\end{figure}

\subsection{The hierarchical operator}

Applying the operator to a vector $\mathbf{x}$ proceeds in three steps
at each layer~$k$:
\begin{enumerate}
\item \textbf{Aggregate:} $\mathbf{x}_k = G_k \, \mathbf{x}$
\item \textbf{Interact:}
  $\mathbf{c}_k = (F_k - \mathrm{Corr}_k) \, \mathbf{x}_k$
\item \textbf{Expand:}
  $\mathbf{h} = \sum_k G_k^\top \mathbf{c}_k$
\end{enumerate}
This corresponds to the operator
\begin{equation}
  H = \sum_k G_k^\top (F_k - \mathrm{Corr}_k) \, G_k,
  \label{eq:operator-full}
\end{equation}
which requires $O(n \log n)$ operations under the assumptions in
Section~\ref{sec:complexity}. The correction matrices
$\mathrm{Corr}_k$ are needed because a node pair $(i, j)$ may be
represented at multiple layers: if nodes $i$ and $j$ belong to
representatives that are close enough to appear in both a fine layer
and a coarser layer, both $F_k$ and $F_{k'}$ will contain a nonzero
entry for the corresponding representative pair. Simply summing
$G_k^\top F_k G_k$ over all layers would double-count such pairs;
$\mathrm{Corr}_k$ removes this redundancy.
Each $\mathrm{Corr}_k$ is constructed as follows: for each nonzero entry
$(r, s)$ in $F_k$, we look up their representatives at each coarser
layer~$k'$. If the coarser pair $(r', s')$ also has a stored interaction
in $F_{k'}$, we set $\mathrm{Corr}_k[r, s] = F_{k'}[r', s']$, taking
the value from the nearest coarser layer. Because we iterate only over
the nonzero entries of the sparse matrices $F_k$, the total work is
$O(E_{\text{total}} \cdot K)$, where $E_{\text{total}}$ is the total
number of stored interactions across all layers.

The effect is that each layer contributes only the \emph{difference}
between its own interaction value and the next coarser layer's value for
the same pair. These differences telescope across layers, so that only
the finest-layer value survives.

\begin{proposition}[Exactly-once contribution]
\label{prop:exactly-once}
For every node pair $(i, j)$ that appears in at least one layer, let
$\hat{k}$ be the finest layer at which the representative pair
$(\hat{\imath},\, \hat{\jmath}) = \bigl(r_{\hat{k}}[i],\, r_{\hat{k}}[j]\bigr)$
has a stored interaction in $F_{\hat{k}}$. Then $(i, j)$ contributes
exactly $f(c_{\hat{\imath}\hat{\jmath}})$ to $H$---once, at the finest
available resolution.
\end{proposition}

\begin{proof}
Let $\hat{k} < k_1 < \cdots < k_m$ be the layers at which the
representative pair for $(i, j)$ has a stored interaction, ordered from
finest to coarsest. At layers not in this set, the pair is absent from
$F_\ell$ and contributes zero. At each layer in the set, the correction
subtracts the nearest coarser layer's value:
\begin{itemize}
\item at the coarsest layer $k_m$, no coarser layer stores the pair, so
  $\mathrm{Corr}_{k_m} = 0$ and the contribution is $F_{k_m}$;
\item at each finer layer $k_j$ ($j < m$), the correction subtracts the
  value from $k_{j+1}$, giving contribution $F_{k_j} - F_{k_{j+1}}$.
\end{itemize}
Summing across all layers in the set:
\[
  (F_{\hat{k}} - F_{k_1}) + (F_{k_1} - F_{k_2}) + \cdots
  + (F_{k_{m-1}} - F_{k_m}) + F_{k_m} = F_{\hat{k}},
\]
where each $F_\ell$ is shorthand for
$F_\ell\bigl[r_\ell[i],\, r_\ell[j]\bigr]$. The telescoping sum yields
exactly the finest-layer value, contributed once.
\end{proof}

\noindent
While one could look up $\hat{k}$ for each pair individually, the
correction formulation~\eqref{eq:operator-full} decomposes the operator
into sparse matrix products that can be evaluated layer by layer,
enabling efficient matvec without per-pair lookup.

 \section{Complexity Analysis}
\label{sec:complexity}

In this section we summarize the computational work required to build and
apply the hierarchical operator, and to update it after local
changes to the underlying network. We emphasize that these estimates are
\textbf{conditional on structural assumptions} that hold in many spatial
and transport networks, but not in all graphs.

We denote:
\begin{itemize}
\item $n$ = number of nodes (zones),
\item $m$ = number of edges,
\item $K$ = number of hierarchical layers,
\item $W_k$ = average work for a cutoff-limited
  shortest-path search at layer~$k$ (cutoff $\alpha \cdot \rho_k$),
\item $d_{\mathit{avg}}$ = average number of stored neighbors per
  representative (sparsity factor),
\item $\Delta$ = number of origin--destination pairs for which a local
  change improves the shortest path (``impact size'').
\end{itemize}

\subsection{Structural assumptions}

We do not claim worst-case guarantees over all graphs. Our analysis and
design are tailored to \textbf{sparse, roughly spatial
networks}~\cite{barthelemy2011spatial} of the kind found in transport
and spatial interaction modeling. In particular
we assume:
\begin{enumerate}
\item \textbf{Sparse network:} the graph is sparse, with $m = O(n)$,
  and running Dijkstra's algorithm from a single source requires
  $O(m \log n) \approx O(n \log n)$, or better with specialized
  implementations.
\item \textbf{Controlled local growth:} the number of nodes within a
  ball of radius~$r$ around any node grows at most polynomially in~$r$
  (for example, as in road networks embedded in low-dimensional space).
  This alone does not bound coarse-layer search volume; the near-linear
  overall scaling arises because the number of sources $S_k$ shrinks at
  the same rate that per-source work $W_k$ grows, keeping the product
  $S_k \cdot W_k$ roughly constant across layers (see below).
\item \textbf{Bounded representative degree:} at each layer~$k$, the
  number of stored interactions per representative, $d_{\mathit{avg}}$,
  remains bounded or grows slowly with $n$, due to the use of cutoffs
  and the decay of $f$ with distance.
\item \textbf{Logarithmic depth:} the number of layers $K$ grows slowly
  with system size, typically $K = O(\log n)$ when radii $\rho_k$ increase
  geometrically.
\end{enumerate}
Under these assumptions we expect near-linear or $n \log n$ scaling.
The assumptions can fail---for dense graphs, or for networks with
abundant long-range shortcuts---and then the work can be higher.
Importantly, such structures are not mere mathematical artefacts:
hub-and-spoke systems common in transport (public transport, airline,
or micromobility networks) concentrate long-range connectivity in a
few nodes and can violate the controlled-growth assumption. We return
to this limitation in Section~\ref{sec:discussion}.

\subsection{Building the hierarchy}

Constructing the hierarchy involves selecting representatives $R_k$ at
each layer~$k$ and computing shortest-path costs from each
representative to other nodes within $\alpha \cdot \rho_k$.
If $S_k$ is the number of sources at layer~$k$, each requiring
work~$W_k$, the total work is
\[
  T_{\text{hierarchy}} = \sum_k S_k \cdot W_k.
\]
Because $S_k$ decreases as $k$ increases while $W_k$ grows,
$S_k \cdot W_k$ remains roughly constant across layers, yielding
$T_{\text{hierarchy}} \approx O(n \cdot K_{\text{eff}})$,
where $K_{\text{eff}}$ captures the effective number of representative
runs. Empirically, behavior is close to linear in~$n$.

\subsection{Building interactions and corrections}

Transforming costs $C_k$ into interactions $F_k$ and constructing
correction matrices $\mathrm{Corr}_k$ requires iterating over stored
representative pairs. The total number of stored entries is
$E_{\text{total}} = \sum_k E_k = O(n \cdot K \cdot d_{\mathit{avg}})$,
since $E_k \propto |R_k| \cdot d_{\mathit{avg}}$ and $|R_k|$ shrinks
with~$k$.

Correction construction visits each fine-layer entry and searches coarser
layers until a correction is found.  In the worst case each entry is
visited $K$ times, giving
$T_{\text{interaction}} = O(n \cdot K^2 \cdot d_{\mathit{avg}})$;
in practice early termination brings this closer to
$O(n \cdot K \cdot d_{\mathit{avg}})$. With $K \approx \log n$ and
modest $d_{\mathit{avg}}$, the effective scaling is $O(n \log n)$.

\subsection{Applying the operator}

Applying the operator $H$ to a vector $\mathbf{x}$ consists of:
\begin{enumerate}
\item aggregating $\mathbf{x}$ to representatives at each layer:
  $\mathbf{x}_k = G_k \, \mathbf{x}$,
\item computing
  $\mathbf{c}_k = (F_k - \mathrm{Corr}_k) \, \mathbf{x}_k$, and
\item expanding and summing:
  $\mathbf{h} = \sum_k G_k^\top \mathbf{c}_k$.
\end{enumerate}
Aggregation and expansion each require $O(n \cdot K)$ work; the dominant
cost is the sparse matrix--vector multiplies whose total nonzeros are
$E_{\text{total}}$. Thus:
\begin{equation}
  T_{\text{apply}} = O(E_{\text{total}}) + O(n \cdot K)
  \approx O(n \cdot K \cdot d_{\mathit{avg}}).
  \label{eq:apply-cost}
\end{equation}
With $K \approx \log n$ and modest $d_{\mathit{avg}}$,
$T_{\text{apply}} \approx O(n \log n)$ in the transport-like networks
that motivate our work.

\subsection{Dynamic updates}
\label{sec:dynamic-updates}

The hierarchy admits an incremental update framework in which update
propagation is proportional to the number of affected pairs times
the hierarchy depth, rather than to overall network size. A detailed
analysis is provided in the Supplementary Information.

\subsection{Summary}

Under standard sparsity and growth assumptions for spatial and transport
networks, building and applying the operator scale as $O(n \log n)$.
Once affected pairs are identified, update propagation scales with
$\Delta \cdot K$ rather than with $n$
(Supplementary Information). We make no universal worst-case
guarantees; our claims are intended for sparse, geographically embedded
networks.
 \section{Approximation Error}
\label{sec:error}

The hierarchical operator is an approximation to the dense interaction
kernel defined by~\eqref{eq:interaction}, where $c_{ij} = c(i, j)$ is
the exact shortest-path cost between nodes $i$ and $j$. In this section
we outline the main sources of approximation error and how they relate
to design parameters of the hierarchy.

At a high level, errors arise from two choices:
\begin{enumerate}
\item \textbf{Grouping error:} nodes within the same group at a given
  layer~$k$ share a representative $r_k[i]$. When we approximate
  interactions at that layer using representative costs
  $c(r_k[i], r_k[j])$ rather than the exact node-to-node costs
  $c(i, j)$, we incur a discrepancy that depends on how far nodes
  deviate from their representatives.
\item \textbf{Truncation across layers:} interactions at a given scale
  are only stored within a cutoff range $\alpha \cdot \rho_k$. Beyond this
  range, interactions are delegated to coarser layers. This can slightly
  distort how interaction mass is partitioned across scales, but by
  construction each origin--destination pair is still accounted for at
  some layer, so we do not ``drop'' interactions entirely.
\end{enumerate}

A simple way to quantify grouping error is to bound the diameter of
groups. Suppose that at layer~$k$ the maximum distance between any
node~$i$ and its representative $r_k[i]$ is at most $\delta_k$:
\[
  c(i, r_k[i]) \le \delta_k \quad \forall\, i.
\]
If the decay function $f$ is Lipschitz-continuous\footnote{The Lipschitz
assumption effectively requires $f$ to have no singularity at zero cost;
all standard decay functions with a positive offset satisfy it.} with
constant $L_f$, i.e.\ $|f(x) - f(y)| \le L_f \cdot |x - y|$, then for
any node pair
$(i, j)$ represented at layer~$k$ we have
\begin{equation}
  |f(c(i,j)) - f(c(r_k[i], r_k[j]))| \le L_f \cdot
  |c(i,j) - c(r_k[i], r_k[j])|.
  \label{eq:lipschitz-bound}
\end{equation}
A concrete bound on the distance discrepancy follows from the triangle
inequality on the shortest-path metric. Since $c$ satisfies the
triangle inequality on undirected graphs with non-negative edge costs:
\begin{equation}
  |c(i,j) - c(r_k[i], r_k[j])| \le c(i, r_k[i]) + c(j, r_k[j])
  \le 2\delta_k.
  \label{eq:triangle-bound}
\end{equation}
Combining~\eqref{eq:lipschitz-bound} and~\eqref{eq:triangle-bound}
gives a per-entry error bound:
\begin{equation}
  |f(c(i,j)) - f(c(r_k[i], r_k[j]))| \le 2\,L_f\,\delta_k.
  \label{eq:combined-bound}
\end{equation}
This bound is simple but always valid: the per-entry interaction error
at layer~$k$ is controlled by twice the product of the Lipschitz
constant of $f$ and the maximum group radius $\delta_k$.

\paragraph{Near-field and far-field error regimes.}
For power-law decay $f(c) = (c + \delta_0)^{-\beta}$ the global
Lipschitz constant $L_f = \beta\,\delta_0^{-\beta-1}$ can be very
large, making~\eqref{eq:combined-bound} appear loose. The hierarchical
structure resolves this naturally. At the finest layer ($k = 0$) every
node represents itself, so pairs resolved there incur \emph{zero}
grouping error. Pairs resolved at layer $k \ge 1$ were not within the
cutoff at layer~$k{-}1$, so their true costs exceed
$\alpha \cdot \rho_{k-1} - 2\delta_{k-1}$. In this range the relevant
Lipschitz constant is not the global one but the maximum slope of $f$
over large costs, which decreases rapidly as the decay function
flattens. The bound~\eqref{eq:combined-bound} thus splits into a
near-field regime (exact) and a far-field regime (bounded by the
\emph{local} smoothness of $f$). For exponential decay the global
Lipschitz constant $1/\lambda$ is already bounded, and the distinction
is less critical.

\subsection{Practical guidelines and empirical characterization}

The bounds above yield two practical guidelines:
\begin{itemize}
\item Strongly distance-sensitive kernels (for example with very steep
  decay) have smaller effective $L_f$ in the relevant cost range,
  requiring less aggressive hierarchy parameters to keep errors low.
\item Increasing the overlap factor~$\alpha$ and choosing base radius
  $\rho_0$ commensurate with the network's mean edge cost reduces group
  diameters at any given scale, at the cost of more computation and
  memory.
\end{itemize}
The bounds above control the \emph{magnitude} of the error but not its
\emph{sign}. The sign structure matters in practice: because decay
functions are convex, symmetric distance errors translate into
interaction errors with a positive mean (Jensen's inequality), so the
operator systematically overestimates interactions in the far field. A
second-order analysis of this bias is given in the SI, and
Section~\ref{sec:error-anatomy} quantifies it empirically at national
scale.

Because the theoretical bounds involve quantities ($\delta_k$, $L_f$)
that depend on graph geometry, they are most useful as qualitative
guides. We therefore complement them with systematic \textbf{empirical
error characterization}: in Section~\ref{sec:experiments} we compare
$H\mathbf{x}$ against the dense baseline $F\mathbf{x}$ across 1{,}800
parameter configurations on synthetic networks
(error below 2.5\% with $\rho_0 = 3\bar{c}$ and $\alpha = 2$,
SI~Table~S3), and dissect distance-resolved error and bias on the
full-scale GB case study (Section~\ref{sec:error-anatomy}).
 \section{Computational Experiments}
\label{sec:experiments}

We validate the operator's scaling behavior and approximation quality
through systematic benchmarks. All experiments were run on a single
machine with an AMD Ryzen Threadripper PRO 7985WX processor (64~cores,
128~hardware threads) and 754\,GB RAM, using Python~3.12,
NumPy~2.4.2, SciPy~1.17.1, NetworkX~3.6.1, and the
HierX package (v0.1.1). The hierarchical operator
uses SciPy's Dijkstra on CSR matrices with 32 parallel workers;
the matrix--vector product is single-threaded.
Synthetic networks have Euclidean edge costs in meters.

\paragraph{Scaling performance.}
\label{sec:scaling}

We measure build and matvec time for random spatial networks from
$n = 100$ to $n = 100{,}000$ zones on a square of side
$\sqrt{n} \times 5{,}000$\,m. Here $\bar{c}$ denotes the mean edge
cost of the network, a characteristic local spacing scale;
for these networks $\bar{c} \approx 10{,}000$\,m.
Parameters: $\rho_0 = 10{,}000\text{\,m} \approx 1\bar{c}$,
$\gamma = 2$, $\alpha = 1.5$, $f(c) = (c + 1000)^{-2}$.
Fitting $t = a \cdot n \log n$ yields $R^2 = 0.998$ for total
build and $R^2 = 0.995$ for matvec, confirming $O(n \log n)$ scaling
(SI~Table~S1, Figure~\ref{fig:scaling}).
Activity vectors are uniform ($\mathbf{x} = \mathbf{1}$).
At $n = 10{,}000$ the hierarchical method is $30\times$ faster than the
dense $O(n^2)$ baseline.

\begin{figure}[t]
  \centering
  \includegraphics[width=\textwidth]{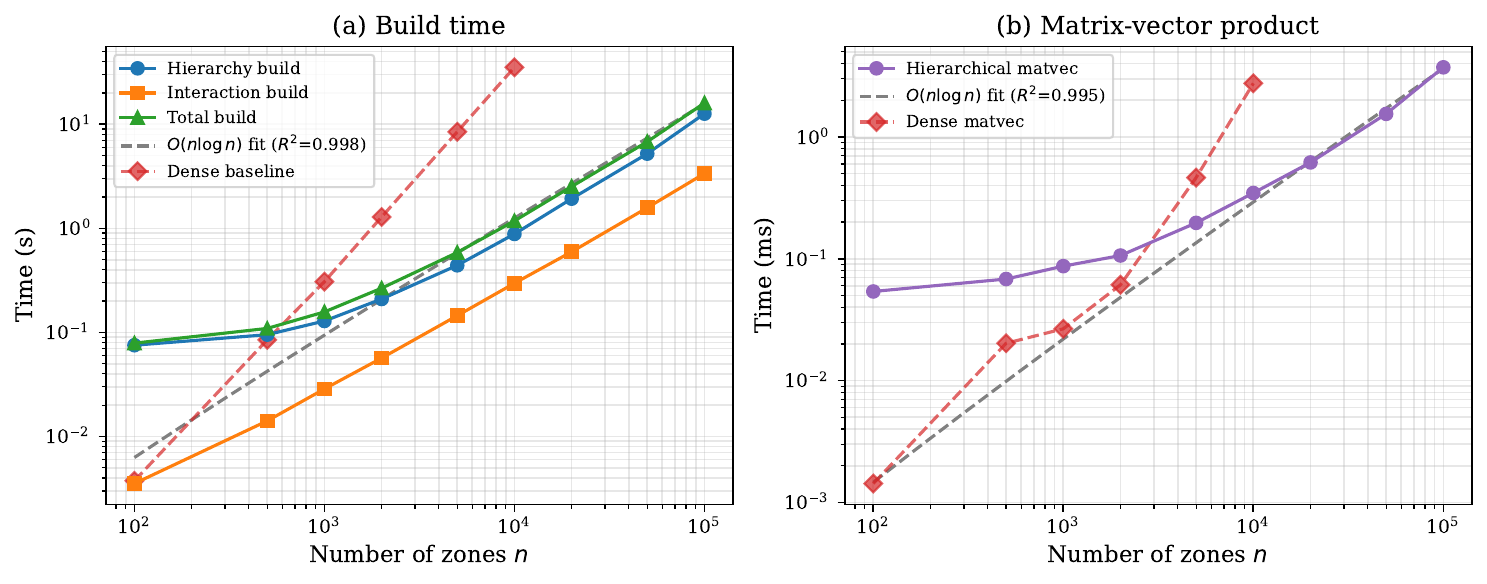}
  \caption{(a)~Build time vs.\ number of zones with $O(n \log n)$ fit.
    (b)~Matvec time with dense baseline comparison.
    \textit{Alt text:} Two log-log plots. Panel~(a) shows hierarchy build
    time and interaction build time versus number of zones from 100 to
    100,000, with an $O(n \log n)$ fit line and a dense baseline that grows
    much faster. Panel~(b) shows matrix-vector product time versus number
    of zones with an $O(n \log n)$ fit and dense baseline comparison.}
  \label{fig:scaling}
\end{figure}

\paragraph{Approximation error vs.\ parameters.}
\label{sec:error-results}

A systematic parameter sweep across 1{,}800 configurations
(5~base radii $\times$ 5~overlap factors $\times$ 3~increase factors
$\times$ 4~interaction functions $\times$ 6~networks) confirms that
the overlap factor $\alpha$ is the primary accuracy knob: increasing
$\alpha$ from 1.0 to 3.0 reduces median error by ${\sim}7\times$
(SI~Table~S2, Figure~\ref{fig:error-overlap}).
The joint dependence on base radius and overlap factor is shown in
SI~Fig.~S1.

\begin{figure}[t]
  \centering
  \includegraphics[width=\linewidth]{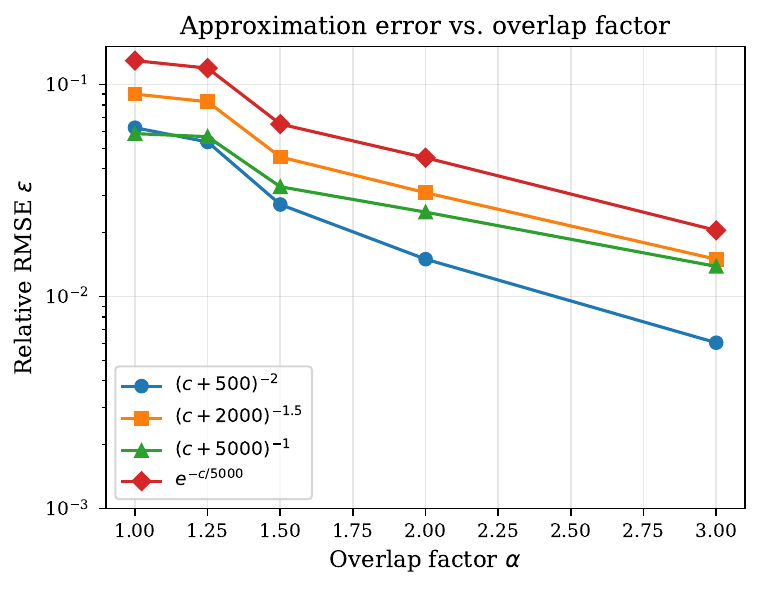}
  \caption{Relative RMSE $\varepsilon$ vs.\ overlap factor for four interaction
    functions (20$\times$20 grid, $\gamma = 2$, $\rho_0 = 4{,}000$\,m).
    \textit{Alt text:} Line plot of relative RMSE versus overlap factor
    for four distance-decay kernels on a 20$\times$20 grid. All four
    curves decrease as overlap factor increases from 1.0 to 3.0, with
    diminishing returns at higher values.}
  \label{fig:error-overlap}
\end{figure}

With $\rho_0 = 3\bar{c}$ ($\bar{c} \approx 10{,}000$\,m,
giving $\rho_0 \approx 30{,}000$\,m) and $\alpha = 2.0$, relative error remains
below 2.5\% up to $n = 5{,}000$ (SI~Fig.~S2, SI~Table~S3).
For the tested kernels, steeper decay generally gives lower error, consistent with
the near-field/far-field analysis of Section~\ref{sec:error}
(SI~Fig.~S3). The hierarchical representation achieves up to
$10\times$ compression in stored entries compared to the dense
matrix (SI~Fig.~S4).

\subsection{Comparison with alternative methods}
\label{sec:baselines}

We compare the hierarchical operator with distance cutoff truncation
and Nystr\"om low-rank approximation on a 25{,}000-zone grid network
($158 \times 158$, spacing 1{,}000\,m, diameter 314\,km).
To compare methods independently of hardware and parallelization
strategy, we report \emph{total nodes explored}: the total number of
finite-distance entries discovered across all Dijkstra calls during
the build phase. For accuracy we report the relative RMSE
$\varepsilon = \|\mathbf{h}^{\mathrm{approx}} - \mathbf{h}^{\mathrm{dense}}\| / \|\mathbf{h}^{\mathrm{dense}}\|$.
Activity vectors are drawn uniformly at random ($x_j \sim U[0.1, 1.1]$) to avoid bias from homogeneous loads.

Distance cutoff computes exact shortest paths within a radius (10\%,
25\%, or 50\% of diameter) and zeros remaining entries.
Nystr\"om samples $m$ landmarks, computes an $n \times m$
cross-similarity matrix~$C$ and approximates the full matrix as
$C\, W^{+} C^{\!\top}$~\cite{williams2001nystrom}, where $W^{+}$ is
the pseudoinverse of the $m \times m$ landmark submatrix. This is
the standard (unmodified) Nystr\"om formulation; indefinite-kernel
extensions exist but were not tested. We use $m = 50$
to $2{,}000$.

Figure~\ref{fig:pareto} shows the error-vs-work Pareto frontier
across four kernels. The hierarchical operator ($\alpha = 2$) explores
approximately 7.0\,M nodes. On the steep kernel $(c+500)^{-2}$, it
achieves 5.1\% RMSE; Nystr\"om fails (${\sim}74\%$ even at
$m = 2{,}000$) because the narrow support defeats sparse landmark
sampling, while a 50\% cutoff achieves 0.7\% but requires $74\times$
more work. On the moderate kernel $(c+2000)^{-1.5}$, the hierarchical
operator gives 6.9\% RMSE; Nystr\"om at $m = 2{,}000$ reaches
comparable error (6.5\%) at $7\times$ the cost. On the shallow kernel
$(c+5000)^{-1}$, Nystr\"om is competitive---5.8\% RMSE at $m = 200$
(5.0\,M nodes) versus the operator's 5.3\%---and achieves 0.4\% at
higher landmark counts, the one regime where its low-rank assumption
holds. On the exponential kernel $e^{-c/5000}$, the operator gives
9.1\% RMSE; cutoff at 10\% diameter achieves 1.1\% at $6\times$ more
work, while Nystr\"om again fails ($> 57\%$ error up to $m = 200$).

\begin{figure}[t]
  \centering
  \includegraphics[width=\textwidth]{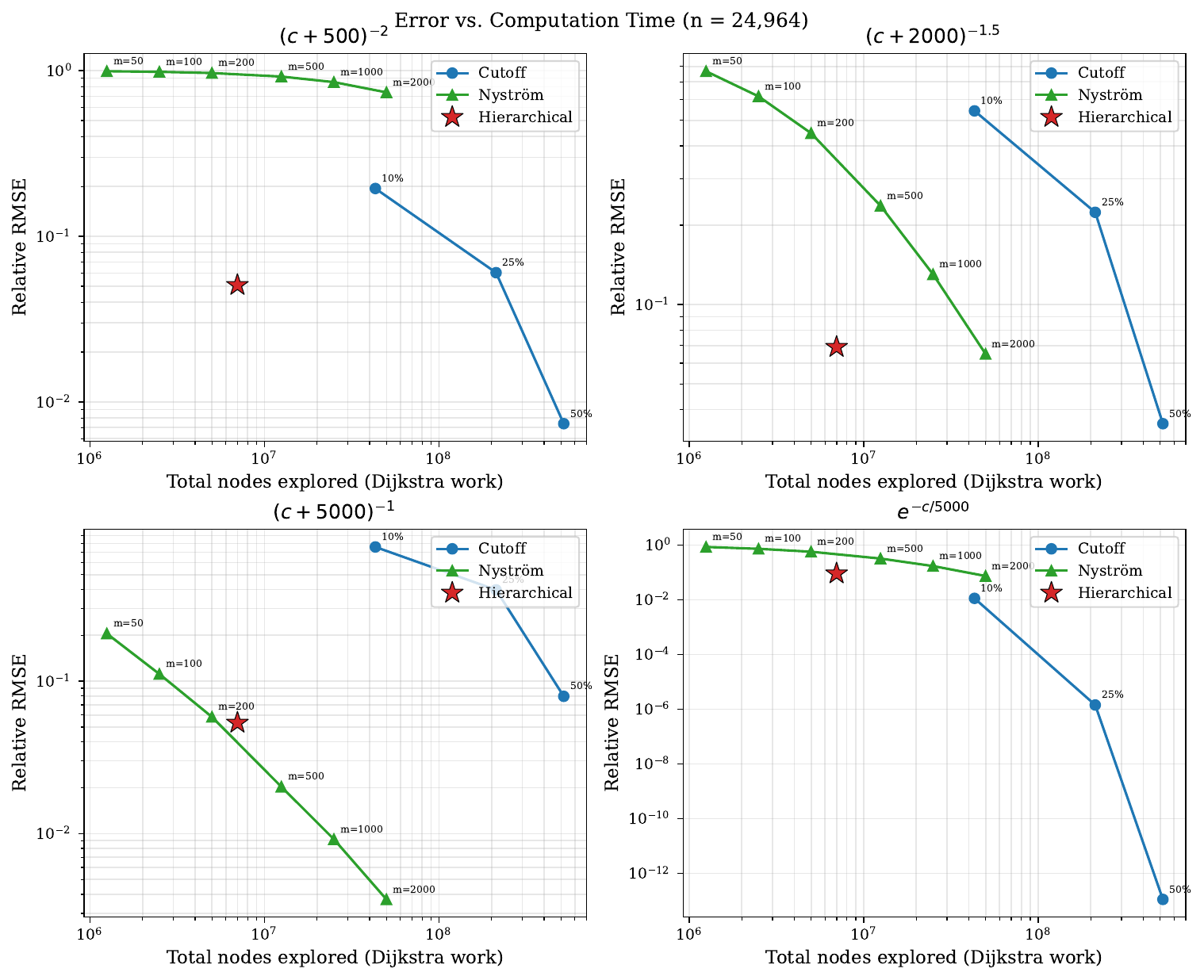}
  \caption{Error-vs-computational-work Pareto frontier on a
    25{,}000-zone grid network (spacing 1{,}000\,m, diameter 314\,km).
    The $x$-axis is total nodes explored (see text for definition);
    the $y$-axis is relative RMSE. Each panel corresponds to one
    interaction function. Cutoff variants are labeled by radius
    fraction of the network diameter; Nystr\"om variants by landmark
    count~$m$. The hierarchical operator (star) achieves competitive accuracy
    on steep kernels; on the smoothest kernel,
    Nystr\"om dominates the error-work frontier.
    \textit{Alt text:} Four Pareto plots on a 25,000-zone grid, one per
    interaction kernel. Each panel shows relative RMSE versus total nodes
    explored for three methods: distance cutoffs (circles), Nystr\"om
    landmarks (triangles), and the hierarchical operator (star). The
    hierarchical method achieves competitive error at substantially lower
    computational cost for steep and moderate kernels.}
  \label{fig:pareto}
\end{figure}

\subsection{Application to real-world networks}
\label{sec:london}

To validate beyond synthetic benchmarks, we apply the operator to
real road networks extracted from OpenStreetMap~\cite{openstreetmap}.
Two London networks are used: a 13\,km-radius drive network centered on
Charing Cross ($n = 61{,}432$ nodes, $\bar{c} = 9.1$\,s, diameter
3{,}181\,s) and a 5\,km-radius walk network ($n = 70{,}095$ nodes,
$\bar{c} = 25.8$\,s).
Edge costs are free-flow travel times (seconds) imputed from OSM speed
limit tags.

On the driving network with steep decay $(c+300)^{-2}$, the
hierarchical operator ($\alpha = 2$, $\rho_0 = 3\bar{c}$) achieves
8.0\% RMSE while exploring only 15.8\,M nodes---${\sim}54\times$
less Dijkstra work than the 25\% cutoff, with $5\times$ better
accuracy. On the very steep kernel $(c+60)^{-2}$, Nystr\"{o}m fails:
even $m = 1{,}000$ landmarks yield 13.5\% RMSE, while the hierarchical
operator achieves 8.6\% RMSE with ${\sim}4\times$ less work.
On broader kernels, Nystr\"{o}m dominates the error-work
tradeoff (SI~Fig.~S5); the hierarchical operator's advantage is
confined to steep decay functions where low-rank approximation breaks
down.

The Nystr\"{o}m failure is even more pronounced on the walking
network with $(c+60)^{-2}$, where the 60\,s offset
(${\approx}\,83$\,m at walking speed) produces steep local
gradients that sparse landmark sampling cannot capture.
Even $m = 1{,}000$ landmarks yield 50.2\% RMSE.
Distance cutoff at 50\% diameter achieves 1.8\% RMSE but explores
3.9\,B nodes. The hierarchical operator at
$\alpha = 5$ achieves comparable accuracy (2.5\% RMSE) while
exploring only 91.9\,M nodes---$43\times$ less Dijkstra work. Figure~\ref{fig:london-combined} summarizes
the London results.

\begin{figure}[t]
  \centering
  \includegraphics[width=\textwidth]{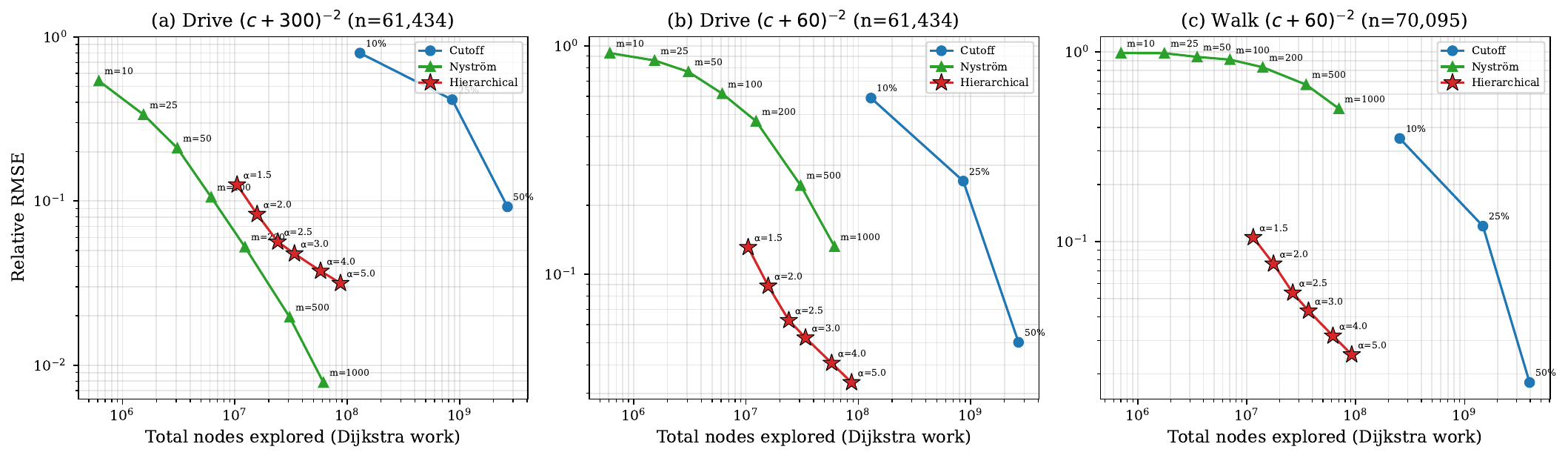}
  \caption{Error-vs-computational-work tradeoff on London networks.
    (a)~Drive network ($n = 61{,}432$), steep decay $(c+300)^{-2}$.
    (b)~Drive network, very steep decay $(c+60)^{-2}$.
    (c)~Walking network ($n = 70{,}095$), steep decay $(c+60)^{-2}$.
    The hierarchical operator (stars) achieves competitive accuracy at
    substantially lower computational cost than distance cutoffs (circles).
    Nystr\"{o}m (triangles) fails on the steep kernels.
    Full results for all four drive-network kernels are in SI~Fig.~S5.
    \textit{Alt text:} Three Pareto plots comparing methods on London
    networks. Panel~(a) shows the 61,432-node driving network with steep
    decay, panel~(b) shows the same network with very steep decay, and
    panel~(c) shows the 70,095-node walking network with steep decay. In
    all panels the hierarchical operator achieves low error at much lower
    computational cost than distance cutoffs. Nystr\"om approximation
    degrades severely on steeper kernels.}
  \label{fig:london-combined}
\end{figure}

\paragraph{Accessibility at metropolitan and national scales.}
\label{sec:pop-accessibility}

To demonstrate the operator at scale, we compute
population-weighted accessibility on the Great Britain driving network
(2{,}576{,}491~nodes, 3{,}017{,}232~edges) and population- and
employment-weighted accessibility on the London pedestrian network
(1{,}765{,}539 nodes, 1{,}917{,}400 edges).
Activity vectors are constructed by disaggregating Census Output
Area population counts~\cite{ons2023census,nrs2024census} to individual
building footprints in proportion to estimated floor area. For the
London walking network, each building is snapped to its nearest network
edge and its count split equally between the two endpoint nodes; for the
GB driving network, each building is snapped to its nearest network
node. Of the 65.0~million residents in the combined GB census,
59.7~million (91.7\%) are successfully assigned to network nodes;
the remainder reside in areas without matched OSM building footprints.

For the GB network, we use $f(c) = (c+300)^{-2}$
with offset matching a 5-minute drive. The hierarchy (10~layers,
$\rho_0 = 60$\,s, $\alpha = 2$) builds in ${\sim}4{,}300$\,s and
stores 720~million interaction entries; each matvec completes in
${\sim}700$\,ms. For the London walking network (7.7~million
residents, 4.2~million workplace population), a steeper decay
$f(c) = (c+60)^{-2}$ captures the shorter range of pedestrian
interaction. The hierarchy (11~layers, $\rho_0 = 3\bar{c}$,
$\alpha = 2$) builds in ${\sim}2{,}000$\,s and stores 82~million
interaction entries; each matvec completes in ${\sim}150$\,ms.
Computing dense ground truth for the full fields is impractical at
these scales; accuracy is validated on
the 61{,}432- and 70{,}095-node subgraphs (Section~\ref{sec:london})
and, for the GB network, directly at full scale via sampled exact
ground truth (Section~\ref{sec:error-anatomy}).
In both cases, a single hierarchy build supports arbitrary activity
vectors at marginal cost per evaluation.

Figure~\ref{fig:pop-accessibility} shows the resulting maps.
Panel~(a) reveals the expected national pattern: major urban
agglomerations appear as bright clusters, with accessibility tapering
into rural and highland areas. Panels~(b) and~(c) compare population
and employment accessibility on the London walking network: population
accessibility reflects London's distributed residential fabric, while
employment accessibility peaks sharply in the City and West End.

\begin{figure}[t]
  \centering
  \includegraphics[width=\textwidth]{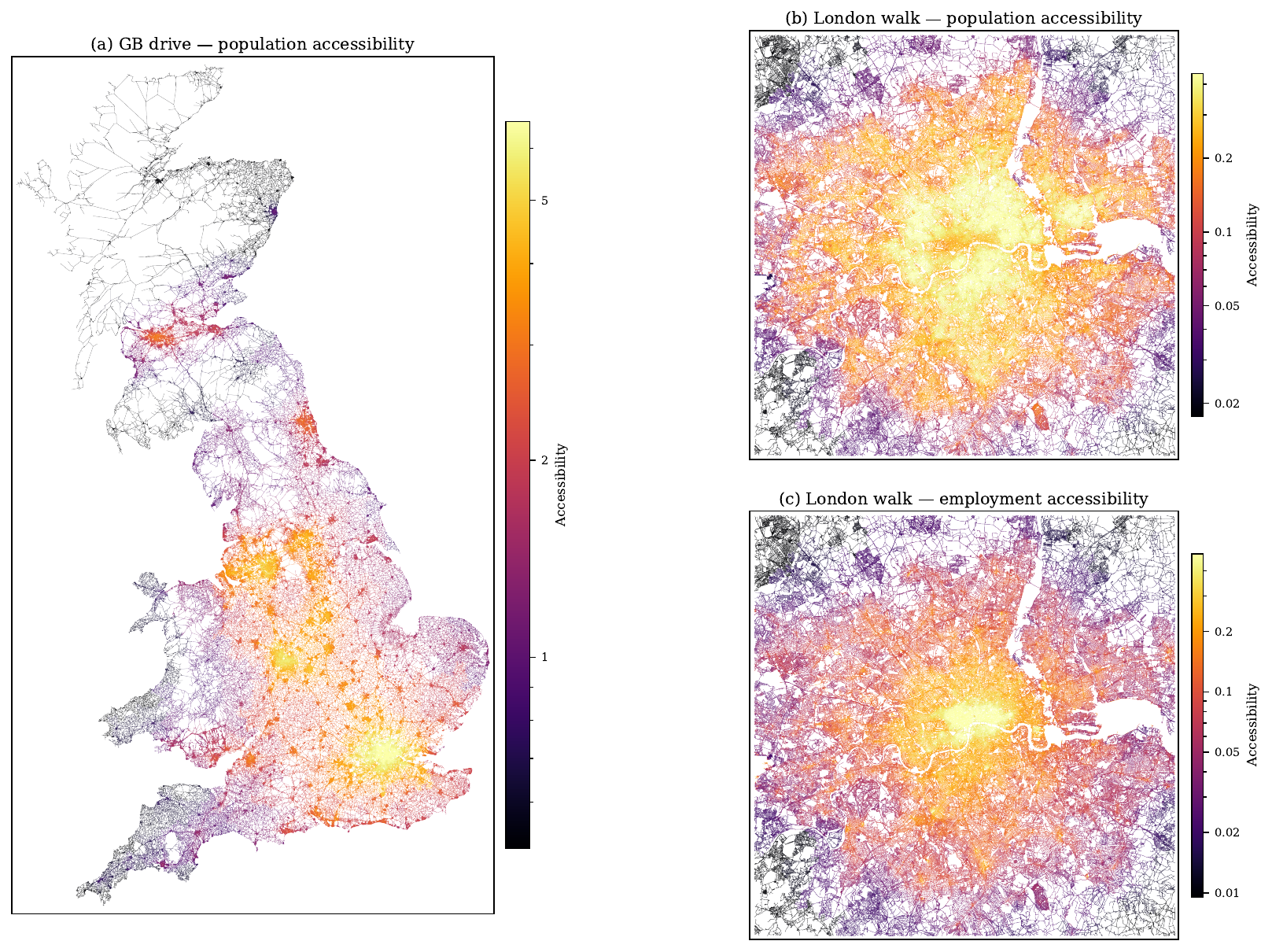}
  \caption{Accessibility maps computed by the hierarchical operator.
    (a)~Great Britain driving network ($n = 2{,}576{,}491$):
    population-weighted accessibility with $f(c) = (c+300)^{-2}$.
    (b,\,c)~London pedestrian network ($n = 1{,}765{,}539$):
    population and employment accessibility with $f(c) = (c+60)^{-2}$,
    walking speed 5\,km/h.
    Edges are encoded by the mean accessibility of their endpoint nodes
    (log scale, per-panel normalization). Activity vectors are Census
    counts disaggregated to buildings by floor area. Each hierarchy
    is built once; subsequent matvec evaluations take ${\sim}700$\,ms
    (GB) and ${\sim}150$\,ms (London).
    An interactive zoomable version is available at
    \url{https://hierx.github.io/hierx-paper/interactive_maps/index.html}.
    \textit{Alt text:} Three network maps showing accessibility computed
    with the hierarchical operator. Panel~(a) shows the Great Britain
    driving network with population-weighted accessibility, revealing high
    values in major urban centers. Panels~(b) and~(c) show the London
    pedestrian network with population and employment accessibility
    respectively; employment accessibility peaks sharply in central London
    while population accessibility is more evenly distributed. A log scale
    is used throughout.}
  \label{fig:pop-accessibility}
\end{figure}

\subsection{Error anatomy of the national-scale case study}
\label{sec:error-anatomy}

While dense ground truth for the full 2.58-million-node field is out of
reach, \emph{exact} values are computable at sampled nodes: each exact
accessibility value $h_i$ requires only one full Dijkstra run from
node~$i$. We exploit this to dissect, at full national scale, how the
operator's error depends on distance, on the decay kernel, and on the
activity distribution---and whether the error is systematically biased.
We sample 1{,}000 origins uniformly at random, compute exact network
distances to all nodes, and compare against the hierarchical effective
cost (the finest-layer representative cost of
Proposition~\ref{prop:exactly-once}) for 536{,}656 origin--destination
pairs stratified across log-spaced distance bins. The hierarchy is
rebuilt deterministically and reproduces the
Figure~\ref{fig:pop-accessibility}a accessibility field exactly, so
these statistics characterize that map directly.

\textbf{Distance error.} Pairs within the finest-layer cutoff
(${\lesssim}\,2$\,min) are resolved exactly
(Figure~\ref{fig:error-anatomy}a). Beyond it, the relative distance
error is nearly scale-free: per-bin RMSE stays at ${\sim}20\%$ across
three orders of magnitude of distance, with a small systematic
underestimate (mean $-2.6\%$). Because layer radii grow geometrically
with distance, the \emph{absolute} error grows proportionally to
distance while the \emph{relative} error remains roughly constant.

\textbf{Interaction bias.} These distance errors translate into
interaction errors asymmetrically: decay kernels are convex, so by
Jensen's inequality even unbiased distance errors inflate the expected
interaction. A second-order expansion for $f(c) = (c+\delta_0)^{-\beta}$
at $c \gg \delta_0$ predicts a mean relative interaction error of
${\approx}\,\tfrac{1}{2}\beta(\beta+1)\sigma^2 - \beta\mu$, where
$\sigma$ and $\mu$ are the standard deviation and mean of the relative
distance error in the far field (pairs beyond the finest-layer cutoff
$\alpha\rho_0$, where grouping first occurs; SI). With
$\sigma \approx 0.20$ and $\mu \approx -0.027$, this predicts biases
from ${+}6.5\%$ ($\beta = 1$) to ${+}17\%$ ($\beta = 2$), in good
agreement with the
observed far-field bias (Figure~\ref{fig:error-anatomy}b) and with the
accessibility-level bias at the sampled nodes: ${+}3.7\%$ for
$(c+3600)^{-1}$, ${+}7.2\%$ for $(c+1800)^{-1.5}$, ${+}14.9\%$ for the
map kernel $(c+300)^{-2}$, and ${+}16.5\%$ for $(c+60)^{-2}$
(SI~Table~S4). Note the reversal relative to small networks: when the
kernel's interaction mass fits within the finest exact layer
($\alpha\rho_0$), steeper kernels are \emph{more} accurate
(Section~\ref{sec:error-results}); at national scale the offsets place
most mass beyond that cutoff, where convexity dominates, and steeper
kernels are \emph{less} accurate. Because these per-pair errors enter
the accessibility sum weighted by interaction mass, the large
far-field values in Figure~\ref{fig:error-anatomy}b contribute
little---for the map kernel, pairs beyond 230\,min carry only
${\sim}1\%$ of the mass.

The total error of the GB population-accessibility field
(Figure~\ref{fig:pop-accessibility}a) at
the sampled nodes is $\varepsilon = 16.4\%$, and it decomposes into a
largely uniform multiplicative inflation (total interaction mass ratio
1.139) plus dispersion: calibrating away the global factor reduces the
error to 8.7\%, and the rank correlation between approximate and exact
accessibility is $\rho = 0.98$. The systematic component is thus
predictable in direction and magnitude, and applications sensitive to
absolute levels should calibrate it away or increase~$\alpha$; the
spatial pattern of the map is robust.

\textbf{Activity sparsity.} For activity concentrated on $k$ support
nodes, exact ground truth for the \emph{entire} network requires only
$k$ Dijkstra runs. Sweeping $k$ with population-weighted support
sampling (Figure~\ref{fig:error-anatomy}c, SI~Table~S5), the
full-network RMSE falls from 23.7\% at $k = 10$ to 16.1\% at
$k = 1{,}000$, approaching the 16.4\% of the ubiquitous census
population vector ($59.7$\,M persons spread over $1.77$\,M nodes),
while the bias stays at ${\approx}{+}15\%$ throughout. Sparse activity
layers---accessibility to a small set of destinations such as
airports, ports, or major hospitals---thus add variance, not bias:
error concentrates where few support nodes are nearby, as
their far-field contributions pass through coarse layers, while the
systematic component is a property of the kernel and the distance
distribution alone.
When the important locations are known in advance, this residual is
avoidable: resolving them at the finest layer---storing their exact
interactions directly, which the correction matrices reconcile with the
coarser layers by construction
(Proposition~\ref{prop:exactly-once})---makes their contributions exact,
at a cost of one shortest-path solve and one dense row per designated
location.

\begin{figure}[t]
  \centering
  \includegraphics[width=\textwidth]{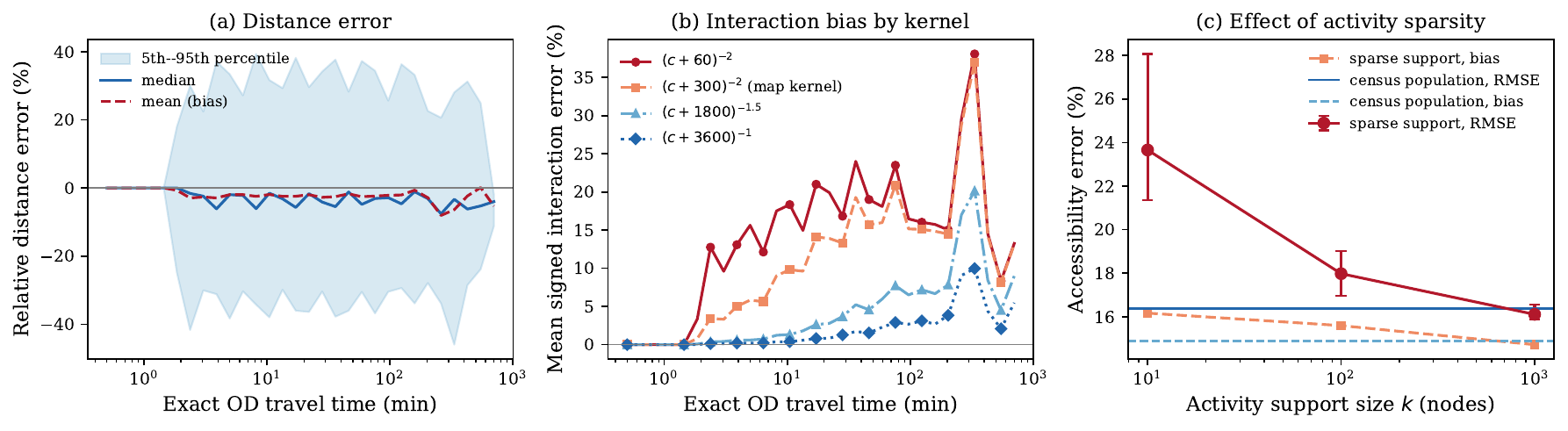}
  \caption{Error anatomy of the GB case study ($n = 2{,}576{,}491$,
    $\alpha = 2$, $\rho_0 = 60$\,s), against exact Dijkstra ground
    truth at 1{,}000 sampled origins.
    (a)~Relative distance error vs.\ exact OD travel time
    (536{,}656 sampled pairs): exact in the near field, then a
    nearly scale-free ${\sim}20\%$ spread with a small negative bias.
    (b)~Mean signed interaction error by distance bin for four kernels:
    convexity turns distance spread into systematically positive
    interaction bias, increasing with kernel steepness.
    (c)~Accessibility error vs.\ activity support size~$k$ (unit
    activity on $k$ population-weighted nodes; exact full-network
    ground truth from $k$ Dijkstra runs; whiskers span 3 trials):
    sparsity adds variance while the bias matches the census-population
    level (horizontal lines).
    \textit{Alt text:} Three panels. Panel~(a) shows relative distance
    error versus exact travel time on a log axis: zero error below two
    minutes, then a roughly constant band of plus or minus twenty
    percent with the mean slightly below zero. Panel~(b) shows mean
    signed interaction error versus travel time for four decay kernels:
    all curves are positive beyond the near field and steeper kernels
    have larger bias, exceeding thirty percent at the longest
    distances. Panel~(c) shows accessibility error versus activity
    support size: root-mean-square error decreases from about
    twenty-four percent at ten support nodes toward the census
    population reference near sixteen percent, while bias stays near
    fifteen percent for all support sizes.}
  \label{fig:error-anatomy}
\end{figure}
 \section{Discussion}
\label{sec:discussion}

The hierarchical operator approximates global interactions on large
networks by exploiting the same multi-scale structure that makes
these systems tractable in practice: local detail where it matters,
coarser resolution where it does not. Evaluation runs in quasi-linear
time; the hierarchy suggests an incremental update framework in which
update propagation is proportional to the number of affected pairs
$\Delta$ times the hierarchy depth $K$
(Supplementary Information).

The approach targets a specific problem class---shortest-path decay
kernels on general graphs---common in spatial interaction modeling but
not directly addressed by existing hierarchical matrix or fast multipole
methods. In settings with dense connectivity, rapidly changing global
shortcuts, or kernels that do not decay with distance, other techniques
may be more appropriate.

\textbf{When is the hierarchical operator the right choice?}
Simple distance cutoffs suffice when the decay function is steep enough
that 99\% of interaction mass lies within a modest radius, and for small
networks ($n < 10^3$) dense computation is fast and exact. Efficient
shortest-path query structures~\cite{geisberger2012exact,bast2007transit,abraham2011hub}
accelerate individual lookups but do not directly provide the aggregated
vector $\mathbf{h} = H\mathbf{x}$. Nystr\"om
approximation~\cite{williams2001nystrom} and random Fourier
features~\cite{rahimi2007random} assume positive semi-definite kernels,
which shortest-path interaction kernels on general graphs generally
are not. The hierarchical operator is most advantageous when:
(i)~the network is large ($n \gtrsim 10^4$);
(ii)~the decay function is shallow enough that distant interactions
carry meaningful mass;
(iii)~the application requires repeated evaluations of $H\mathbf{x}$
for varying $\mathbf{x}$, amortizing build cost; and
(iv)~the network may evolve, making precomputed static structures
expensive to maintain.
Our baseline comparison (Section~\ref{sec:baselines}) confirms this:
the hierarchical operator achieves 5--9\% RMSE; matching this accuracy
with distance cutoffs requires ${\sim}30\times$ more Dijkstra work,
and exceeding it can cost up to $74\times$ more. Nystr\"om approximation is competitive only for
the smoothest kernels; for steep decay it exhibits ${\sim}74\%$ RMSE
even at $m = 2{,}000$ landmarks, a limitation compounded by network
interventions that move the distance structure further from any
Euclidean embedding.

\textbf{Parameter sensitivity.} The overlap factor $\alpha$ is the most
important tuning parameter, with a ${\sim}7\times$ error reduction from
$\alpha = 1.0$ to $\alpha = 3.0$ (SI~Table~S2). The base radius
$\rho_0$ should be at least $3\times$ the mean edge cost;
with $\rho_0 = 3\bar{c}$ and $\alpha = 2$, error remains below 2.5\%
up to $n = 5{,}000$ (SI~Table~S3). The increase factor $\gamma = 2$ balances resolution and
computational cost.

\textbf{Topology effects.} Regular grids show the most predictable
behavior; random spatial networks exhibit higher error variance. The
London experiments (Section~\ref{sec:london}) confirm that the operator
handles organic, irregular topology well, achieving 8\% RMSE at
$\alpha = 2$ with ${\sim}54\times$ less Dijkstra work than a 25\%
cutoff. Practitioners should calibrate against dense baselines on
representative subnetworks before full-scale deployment.

\textbf{Interaction function dependence and bias.} The effect of kernel
steepness is scale-dependent. On small and moderate networks, steeper
decay functions produce \emph{lower} error, because interactions
concentrate in nearby zones captured exactly by the finest layer. At
national scale the ordering reverses
(Section~\ref{sec:error-anatomy}): the error budget is carried by
far-field pairs, where the convexity of the decay function converts
the roughly symmetric distance error into a systematically
\emph{positive} interaction bias that grows with steepness
(${\approx}\,\tfrac{1}{2}\beta(\beta+1)\sigma^2$ for power-law
exponent~$\beta$ and relative distance spread~$\sigma$). Practitioners
should expect the operator to overestimate accessibility on average,
by a predictable margin that can be calibrated away via the total
interaction mass, and should increase~$\alpha$ when absolute levels
matter.

\textbf{Hub-dominated and directed networks.} The structural
assumptions of Section~\ref{sec:complexity} can be violated by networks
that are entirely realistic in transport analysis, not only by
mathematical artefacts. Hub-and-spoke systems---public transport with
express services, airline networks, micromobility with docking
stations---concentrate long-range connectivity in a few nodes, so that
the ball of reachable nodes around a hub grows much faster than
polynomially in radius. This inflates the per-source search work $W_k$
and the stored degree $d_{\mathit{avg}}$, eroding the complexity
advantage, and shortcut-rich metrics also blur the separation of
scales on which the grouping rests. A partial mitigation is that the
hierarchy operates on the network metric itself rather than on a
Euclidean embedding, so it degrades gracefully where Euclidean-based
methods break outright; hub-aware representative selection (placing
representatives at hubs) is an open direction. Directed graphs
(one-way streets, bus routes) pose a separate, structural limitation:
the current formulation assumes symmetric costs, since grouping and
correction use a single cost per pair. Supporting asymmetric costs
would require separate forward and backward cost structures
(see Supplementary Information).

\textbf{Limitations.} The open-source HierX package does not yet include
dynamic update functionality; the complexity analysis
(Section~\ref{sec:complexity}) establishes an $O(\Delta \cdot K)$
update framework, and preliminary prototype results support
feasibility, but a full empirical validation will be reported
separately. The Python implementation stores intermediate objects that
inflate memory; the intrinsic sparse structure contains ${\sim}10\times$
fewer entries than the dense matrix, a compression advantage a compiled
implementation would fully realize. Additional discussion of directed
graphs, parallelization, and asymmetric costs is provided in the
Supplementary Information.

\textbf{Future directions.} The most pressing next step is implementing
and validating the dynamic update scheme on real infrastructure
scenarios, including the impact-driven propagation theory outlined in
Section~\ref{sec:complexity}. Representative selection could be improved
via graph coarsening or centrality-based criteria. On the build side,
contraction-hierarchy or hub-label queries could replace Dijkstra for
the coarse layers, where few costs between widely separated
representative pairs are needed (one-to-one queries), while
cutoff-limited Dijkstra remains the right tool for the fine layers'
dense one-to-many searches; this division of labor could substantially
reduce build-phase constants. On the accuracy side, the nested group
structure invites $\mathcal{H}^2$-inspired refinements
(Section~\ref{sec:related}): translating cost information between
nested layers instead of recomputing it, and softening the hard
single-representative assignment into a weighted average (convex
combination) of nearby representatives, which in variable-order
$\mathcal{H}^2$-matrices
improves approximation orders under certain conditions and could
reduce the far-field grouping error here. Applying the operator to networks
beyond transport---social, biological, or communication
networks---would test the generality of the hierarchical approximation.
 \section{Conclusion}
\label{sec:conclusion}

This work establishes a hierarchical sparse-plus-correction operator as
an alternative to dense distance-decayed interaction matrices on large
sparse graphs. The operator adapts hierarchical interaction ideas from
fast multipole methods, hierarchical matrices, and multi-scale cellular
automata to shortest-path-based interaction kernels on general graphs.

Systematic benchmarks on networks up to 100{,}000 zones confirm
$O(n \log n)$ scaling for both build ($R^2 = 0.998$) and
operator application ($R^2 = 0.995$). Application to real networks---a 2.58-million-node Great Britain
driving network and a 1.77-million-node London walking
network---demonstrates practical national- and metropolitan-scale
performance: a one-time hierarchy build of roughly an hour yields an
operator that evaluates each accessibility field in well under one
second. With $\rho_0 = 3\bar{c}$ and $\alpha = 2$, relative RMSE remains
below 2.5\% up to $n = 5{,}000$ (SI~Table~S3). Head-to-head comparison with distance cutoff truncation and
Nystr\"om low-rank approximation on a 25{,}000-zone network
shows that, for steep and
moderate decay kernels, it achieves 5--9\% RMSE at a fraction of the
Dijkstra work required by either baseline; for the smoothest kernels
Nystr\"om can achieve lower error at proportionally higher cost, but the
hierarchical operator remains the most efficient option at comparable
accuracy levels.

The complexity analysis establishes an incremental update framework in
which update propagation is proportional to the number of affected
pairs $\Delta$ times the hierarchy depth $K$, rather than to overall
network size. A full empirical validation of this framework is deferred
to future work.

The open-source HierX package, archived datasets, and
Docker-based reproducibility pipeline accompany the paper.

\section*{Materials and Methods}

\subsection*{Data and Code Availability}
The HierX Python package implementing the hierarchical
operator is available as open-source software at
\url{https://github.com/hierx/hierx}.
Benchmark scripts, pre-computed results, and data pipelines needed to
reproduce every figure and table are provided in the companion repository at
\url{https://github.com/hierx/hierx-paper}.
A self-contained Docker image allows full reproduction of all figures
and benchmarks without local installation; see the repository README
for instructions.
The London and Great Britain network snapshots,
pre-computed accessibility arrays, and Census-derived node activity
vectors used in Section~\ref{sec:pop-accessibility} are
archived on Zenodo (\url{https://doi.org/10.5281/zenodo.19062193}) for
long-term reproducibility. These files can also be regenerated from
OpenStreetMap~\cite{openstreetmap} using the included scripts, though
OSM data may change over time.
Census 2021 population (TS001) and workplace (WP001) data for England
and Wales are sourced from the Office for National
Statistics~\cite{ons2023census} under the Open Government Licence;
Scotland Census 2022 population data are sourced from National Records
of Scotland~\cite{nrs2024census}.

\subsection*{AI Disclosure}
Portions of the code and manuscript preparation were assisted by
Claude models (Anthropic; Claude Opus 4.6, Claude Opus 4.8, and
Claude Fable 5). The AI
was used for code generation, benchmark scripting, and editorial
suggestions. All scientific content, methodology, and conclusions are
the sole responsibility of the authors.

\section*{Funding}
This work was supported by the Swedish Transport Administration (Trafikverket), grant TRV 2024/31308.

\section*{Author Declarations}
The authors declare no competing interests.
 
\bibliographystyle{plainnat}

\end{document}